\documentclass[11pt]{article}

\usepackage{etex}

\usepackage{fullpage}

\usepackage{amsmath}
\usepackage{amssymb}
\usepackage{amsthm}
\usepackage{array}
\usepackage{bbm}
\usepackage{cancel}
\usepackage{cmap}
\usepackage{enumerate}
\usepackage{enumitem}
\usepackage{fancyhdr}
\usepackage{mathdots}
\usepackage{mathtools}
\usepackage{mathrsfs}
\usepackage{hyperref}
\usepackage{stackrel}
\usepackage{stmaryrd}
\usepackage{tabularx}
\usepackage{tikz}
\usepackage{ctable}
\usepackage{titlesec}
\usepackage{titletoc}
\usepackage{url}
\usepackage{verbatim}
\usepackage{wasysym}
\usepackage{wrapfig}
\usepackage{yhmath}
\usepackage[all,cmtip]{xy}

\usepackage{hyperref}
\usepackage[%
	firstinits=true,
	doi=false,
	url=false,
	isbn=false,
	backend=bibtex,
	style=alphabetic,hyperref]{biblatex}

\usetikzlibrary{calc,trees,positioning,arrows,chains,shapes.geometric,%
    decorations.pathreplacing,decorations.pathmorphing,shapes,%
    matrix,shapes.symbols,shadows,fadings}

\newtheorem{thm}{Theorem}[section]
\newtheorem*{thm*}{Theorem}

\newtheorem*{prb*}{Problem}

\newtheorem*{ax*}{Axiom}

\newtheorem*{clm*}{Claim}

\newtheorem*{conj*}{Conjecture}

\newtheorem{df}[thm]{Definition}
\newtheorem*{df*}{Definition}

\newtheorem*{ex*}{Example}

\newtheorem{lem}[thm]{Lemma}
\newtheorem*{lem*}{Lemma}

\newtheorem*{pos*}{Postulate}
\newtheorem{pr}[thm]{Proposition}
\newtheorem*{pr*}{Proposition}

\newtheorem*{qu*}{Question}

\newtheorem*{rem*}{Remark}

\newcommand{\E}[0]{\mathbb{E}}
\newcommand{\EE}[0]{\mathop{\mathbb E}}
\newcommand{\F}[0]{\mathbb{F}}

\newcommand{\Pj}[0]{\mathbb{P}}

\newcommand{\R}[0]{\mathbb{R}}

\newcommand{\Z}[0]{\mathbb{Z}}

\newcommand{\mv}[0]{\mathbf{v}}
\newcommand{\mx}[0]{\mathbf{x}}
\newcommand{\my}[0]{\mathbf{y}}

\newcommand{\al}[0]{\alpha}
\newcommand{\be}[0]{\beta}

\newcommand{\de}[0]{\delta}
\newcommand{\De}[0]{\Delta}
\newcommand{\ep}[0]{\varepsilon}

\newcommand{\om}[0]{\omega}
\newcommand{\Om}[0]{\Omega}
\newcommand{\si}[0]{\sigma}

\newcommand{\nin}[0]{\not\in}
\newcommand{\opl}[0]{\oplus}

\newcommand{\subeq}[0]{\subseteq}

\newcommand{\nequiv}[0]{\not\equiv}
\newcommand{\bs}[0]{\backslash}
\newcommand{\iy}[0]{\infty}

\newcommand{\rc}[1]{\frac{1}{#1}}
\newcommand{\prc}[1]{\pa{\rc{#1}}}
\newcommand{\ff}[2]{\left\lfloor\frac{#1}{#2}\right\rfloor}

\newcommand{\fc}[2]{\frac{#1}{#2}}

\newcommand{\pf}[2]{\pa{\frac{#1}{#2}}}

\newcommand{\af}[2]{\ab{\fc{#1}{#2}}}

\newcommand{\ab}[1]{\left| {#1} \right|}
\newcommand{\an}[1]{\left\langle {#1}\right\rangle}

\newcommand{\ce}[1]{\left\lceil {#1}\right\rceil}

\newcommand{\pa}[1]{\left( {#1} \right)}

\newcommand{\ve}[1]{\left\Vert {#1}\right\Vert}

\newcommand{\set}[2]{\left\{{#1}:{#2}\right\}}

\newcommand{\wt}[1]{\widetilde{#1}}

\newcommand{\bwoc}[0]{by way of contradiction}

\newcommand{\Wog}[0]{ Without loss of generality }

\newcommand{\diag}{\operatorname{diag}}

\newcommand{\poly}{\operatorname{poly}}

\newcommand{\rank}{\operatorname{rank}}

\newcommand{\Supp}{\operatorname{Supp}}

\newcommand{\pull}[9]{
#1\ar@/_/[ddr]_{#2} \ar@{.>}[rd]^{#3} \ar@/^/[rrd]^{#4} & &\\
& #5\ar[r]^{#6}\ar[d]^{#8} &#7\ar[d]^{#9} \\}

\newcommand{\cmp}[9]{
\xymatrix{
#1 \ar[r]^{#4}{#5} \ar@/_2pc/[rr]^{#8}_{#9} & #2 \ar[r]^{#6}_{#7} & #3
}
}

\newcommand{\ha}[1]{\ar@{^(->}[#1]}
\newcommand{\ls}[1]{\ar@{-}[#1]}
\newcommand{\sj}[1]{\ar@{->>}[#1]}
\newcommand{\aq}[1]{\ar@{=}[#1]}
\newcommand{\acir}[1]{\ar@{}[#1]|-{\textstyle{\circlearrowright}}}
\newcommand{\acil}[1]{\ar@{}[#1]|-{\textstyle{\circlearrowleft}}}
\newcommand{\ard}[1]{\ar@{.>}[#1]}
\newcommand{\mt}[1]{\ar@{|->}[#1]}
\newcommand{\inm}[1]{\ar@{}[#1]|-{\in}}
\newcommand{\inr}{\ar@{}[d]|-{\rotatebox[origin=c]{-90}{$\in$}}}
\newcommand{\inl}{\ar@{}[u]|-{\rotatebox[origin=c]{90}{$\in$}}}

\newcommand{\smatt}[4]{
\left(\begin{smallmatrix} 
{#1}&{#2}\\
{#3}&{#4}
\end{smallmatrix}\right)
}

\newcommand{\beq}[1]{\begin{equation}\llabel{#1}}
\newcommand{\eeq}[0]{\end{equation}}
\newcommand{\bal}[0]{\begin{align*}}
\newcommand{\eal}[0]{\end{align*}}
\newcommand{\ban}[0]{\begin{align}}
\newcommand{\ean}[0]{\end{align}}

\newcommand{\md}[1]{\,(\text{mod }#1)}

\newcommand{\NP}[0]{\operatorname{\mathsf{NP}}}

\newcommand{\Maj}[0]{\operatorname{\mathsf{Maj}}}

\newcommand{\fixme}[1]{{\color{red}#1}}
\newcommand{\llabel}[1]{\label{#1}\text{\fixme{\tiny#1}}}

\newcommand{\arxiv}[1]{\url{http://www.arxiv.org/abs/#1}}

\newcommand{\vocab}[1]{\textbf{#1}} 

\allowdisplaybreaks[2]

\DeclareFontFamily{U}{wncy}{}
    \DeclareFontShape{U}{wncy}{m}{n}{<->wncyr10}{}
    \DeclareSymbolFont{mcy}{U}{wncy}{m}{n}
    \DeclareMathSymbol{\Sh}{\mathord}{mcy}{"58} 
\begin{document}

\title{Quadratic polynomials of small modulus cannot represent OR}

\author{Holden Lee\thanks{Department of Mathematics,
Princeton University.
Email: \texttt{holdenl@math.princeton.edu}. }}

\date{\today}
\maketitle

\begin{abstract}


An open problem in complexity theory is to find the minimal degree of a polynomial representing the $n$-bit OR function modulo composite $m$. This problem is related to understanding the power of circuits with $\text{MOD}_m$ gates where $m$ is composite. The OR function is of particular interest because 
it is the simplest function not amenable to bounds from communication complexity. 
Tardos and Barrington~\cite{TB} established a lower bound of $\Om((\log n)^{O_m(1)})$, and Barrington, Beigel, and Rudich~\cite{BBR} established an upper bound of $n^{O_m(1)}$. No progress has been made on closing this gap for twenty years, and progress will likely require new techniques~\cite{BL}. 

We make progress on this question viewed from a different perspective: rather than fixing the modulus $m$ and bounding the minimum degree $d$ in terms of the number of variables $n$, we fix the degree $d$ and bound $n$ in terms of the modulus $m$. 
For degree $d=2$, we prove a quasipolynomial bound of $n\le m^{O(d)}\le m^{O(\log m)}$, improving the previous best bound of $2^{O(m)}$ implied by Tardos and Barrington's general bound.

To understand the computational power of quadratic polynomials modulo $m$, we introduce a certain dichotomy which may be of independent interest. Namely, we define a notion of boolean rank of a quadratic polynomial $f$ and relate it to the notion of diagonal rigidity. Using additive combinatorics, we show that when the rank is low, $f(\mx)=0$ must have many solutions. Using techniques from exponential sums, we show that when the rank of $f$ is high, $f$ is close to equidistributed. In either case, $f$ cannot represent the OR function in many variables.
\end{abstract}


\section{Introduction}

\subsection{Overview}

A major open problem in complexity theory is to characterize the computational power of modular counting. For instance, for any composite $m$, the question $\NP\stackrel{?}{\subeq}\mathsf{AC}^0[m]$ is still open, where $\mathsf{AC}^0[m]$ is the class of functions computable by constant-depth circuits allowing $\text{MOD}_m$ gates.

One technique to tackle such problems is to relate circuits containing $\text{MOD}_m$ gates to polynomials over $\Z_m$. This has been successful when $m$ is prime. For example, to show $\text{MOD}_q\nin \mathsf{ACC}^0[p]$ for $p$ prime and any $q$ not a power of $p$, Razborov and Smolensky \cite{R,S} showed that functions in $\mathsf{AC}^0[p]$ can be approximated by polynomials of degree $(\log n)^{O(1)}$, and then proved that $\text{MOD}_m$ cannot be approximated by such polynomials. See \cite{B93} for a survey of the polynomial method in circuit complexity. (See also \cite{V}.)
What if we allow arbitrary moduli? Building on work of Yao~\cite{Y}, Beigel and Tarui \cite{BT} show that functions $f_n$ in $\mathsf{ACC}^0$ can be written in the form $h_n\circ p_n$ where $p_n$ is a polynomial over $\Z$ of degree $(\log n)^{O(1)}$ and $h_n:\Z\to \{0,1\}$ is some function. Thus, to show an explicit family of functions $f_n$ is not in $\mathsf{ACC}^0$, it suffices to lower-bound the minimum degree of polynomials representing $f_n$ in this way. However, currently there are few techniques for doing so.

As a first step towards such lower bounds, Barrington, Beigel, and Rudich~\cite{BBR} consider a similar question over $\Z_m$ rather than $\Z$.  Write $B=\{0,1\}$ below.
\begin{df}\label{df:1}
Let $g:B^n\to B$ be a function.
A function $f:B^n\to \Z_m$ \vocab{weakly represents} $g$ if there exists a partition $\Z_m = A\cup A^c$ such that 
\bal
g(\mx) =0 & \iff f(\mx) \in A\\ 
g(\mx) =1 & \iff f(\mx) \in A^c.
\end{align*}
Define the \vocab{weak degree} $\De(g,m)$ to be the minimal degree of a polynomial $f:B^n\to \Z_m$ that weakly represents $g$.
\end{df}
The goal is to estimate $\De(g,m)$ for specific functions $g$, and in particular exhibit functions $g$ with large weak degree.

One way to bound $\De(g,m)$ is using communication complexity.
Gromulsz \cite{Gr95} noted that if a function has $k$-party communiction complexity $\Om(k)$, then its weak degree is at least $k$.
From Babai, Nisan, and Szegedy's \cite{BNS} lower bound for the communication complexity of the generalized inner product function
he concluded that the GIP function has weak degree $\Om(\log n)$. 
Current techniques in communication complexity only give superconstant bounds when the number of parties is $O(\log n)$~\cite{KN}, so improvement along these lines is difficult.

Researchers have proved bounds for the more rigid notion of 1-sided representation, which requires $A=\{0\}$ in Definition~\ref{df:1}, obtaining bounds of $\Om(N)$ for the equality function $\text{Eq}_N(\mx, \my)$ \cite{KW91} and the majority function $\Maj_N(\mx)$ \cite{Ts93}, and a bound of $N^{\Om(1)}$ for the $\text{MOD}_n,\neg \text{MOD}_n$ when $n$ has a prime not dividing $m$ \cite{BBR}. However, 1-sided representation does not capture the full power of modular counting.

A natural function to consider is the OR function $\text{OR}_n:B^n\to B$, defined by $\text{OR}_n(\mathbf 0)=0$ and $\text{OR}_n(\mx)=1$ for $\mx\ne \mathbf 0$. $\text{OR}_n$ (equivalently $\text{AND}_n$) is a natural function to consider because it is the simplest function, in a sense, and its communication complexity is trivial, so other techniques are necessary to lower bound its degree. Note that because $\text{OR}_n$ takes the value 0 only on $\mathbf 0$, $\De(\text{OR}_n,m)$ is the minimal degree of a polynomial $g$ such that for $\mx\in B^n$, $g(\mx)=0$ iff $\mx=0$ (i.e., weak representation is equivalent to 1-sided representation).


When $m$ is a prime power it is folklore \cite{TB} that
\[
\fc{n}{m-1}\le \De(\text{OR}_n,m)\le n,
\]
because one can turn a polynomial $f$ weakly representing $g$, into a polynomial representing $g$, with at most a $m-1$ factor increase in degree. See also \cite{CFS} for general theorems on the zero sets of polynomials over finite fields.

Most interesting is the regime where $m$ is a fixed composite number (say, 6), and $n\to \iy$. Suppose $m$ has $r$ factors. Barrington, Beigel, and Rudich \cite{BBR} show the upper bound
\[
\De(\text{OR}_n,m) = O(n^{\rc r}).
\]
This bound is attained by a symmetric polynomial. Moreover, they prove that any symmetric polynomial representing $\text{OR}_n$ modulo $m$ has degree $\Om(n^{\rc r})$.

Alon and Beigel~\cite{AB} proved the first superconstant lower bound on the weak degree of $\text{OR}_n$. Later Tardos and Barrington~\cite{TB} proved the bound
\begin{equation}\label{eq:TB}
\De(\text{OR}_n,m)\ge \pa{\pa{\rc{q-1}-o(1)}\log n}^{\rc{r-1}} = \Om_m(\log n)^{\rc{r-1}}
\eeq
where $q$ is the smallest prime power fully dividing $m$. Their proof proceeded by finding a subcube of $B^n$ where the polynomial $f$ is constant modulo a prime power $q$ dividing $m$; then $f$ represents OR modulo $\fc mq$ on this subcube. An induction on the number of distinct prime factors results in the $\rc{r-1}$ exponent. This technique has also been used to show structural theorems of polynomials over $\F_q^n$, with applications to affine and variety extractors~\cite{CT}. 

In this work, we make modest progress on this question. Rather than fixing the modulus $m$ and bounding the minimum degree $d$, we fix the degree $d$ and bound the minimum modulus $m$. Specifically, we focus on the degree 2 case, and prove the following.
\begin{thm}\label{thm:main}
There exists a constant $C$ such that the following holds. 
If $m$ has $d$ prime factors, counted with multiplicity, and the quadratic polynomial $f\in \Z_m[x_1,\ldots, x_n]$ weakly represents $\text{OR}_n$ modulo $m$, then
\[
n\le m^{Cd}\le m^{C\lg m}.
\] 
\end{thm}
The lower bound by Tardos and Barrington~\eqref{eq:TB} gives $n\le q^{2^r}$ where $q$ is the smallest prime power factor of $m$, and $r$ is the number of distinct prime factors. This gives $n \le 2^{\wt O(m)}$. Hence, Theorem~\ref{thm:main} improves this exponential upper bound to a quasipolynomial upper bound.

We conjecture that the correct upper bound is $n= O(m)$, or at the very least, we have $n=O(m^C)$. The $d$ loss comes from an inefficient way of dealing with multiple factors.

To prove Theorem~\ref{thm:main}, we define a new notion of \emph{boolean rank} (Definition~\ref{df:1rank}) for a quadratic polynomial $f$, which differs from the ordinary notion of rank in that it captures rank only over the boolean cube, and has connections to matrix rigidity. This notion of boolean rank enables us to split the proof into two cases that we consider independently. When the rank is low, we use additive combiantorics to show $f(\mx)=0$ must have many solutions. When the rank is high, we use Weyl differencing to show that $f$ is close to equidistributed. In either case, when $m$ is small $f(\mx)=0$ will have more than one solution and hence $f$ cannot represent $\text{OR}_n$.

\paragraph{Organization:} The outline of the rest of the paper is as follows. In the remainder of the introduction, we introduce related work and notations. In Section~\ref{sec:over} we give a more detailed overview of the proof. In Sections~\ref{sec:lowrank} and~\ref{sec:dich} we consider the low and high rank cases, respectively. In Section~\ref{sec:proof} we prove the main theorem. In Section~\ref{sec:thoughts} we speculate on ways to extend the argument to higher degree. Appendix~\ref{sec:app} contains facts we will need about linear algebra over $\Z_m$ when $m$ is composite.

\subsection{Related work}

The problem of finding the weak degree of $\text{OR}_n$ is connected to several other interesting problems. 
Firstly, polynomials representing $\text{OR}_n$ modulo $m$ can be used to construct matching vector families (MVF)~\cite{Gr00}, which can then be used to build constant-query locally decodable codes (LDCs) \cite{E, DGY}. A \emph{matching vector family} modulo $m$ is a pair of lists $s_1,\ldots, s_n,t_1,\ldots, t_n\in \Z_m^n$ such that 
\[
\an{s_i,t_j}\begin{cases}
=0,&i=j\\
\ne 0,&i\ne j.
\end{cases}
\]
If $f$ is a polynomial representing $\text{OR}_n$, then $f((2x_iy_i-x_i-y_i+1)_{1\le i\le n})=0$ iff $\mx = \my$. If this polynomial is $\sum a_{\al, \be}\mx^{\al}\my^{\be}$, then the corresponding MVF consists of the $2^n$ vectors $(a_{\al,\be}\mx^{\al})_{\al,\be},\mx\in B^n$ and $2^n$ vectors  $(\my^{\be})_{\al,\be},\my\in B^n$. 
The representation of $\text{OR}_n$ by symmetric polynomials already gives a subexponential-length LDC. 
There is an large gap between the upper bound and lower bound for constant-query locally decodable codes. For each positive integer $t$, there is a family of constant-query LDCs taking messages of length $n$ to length $\exp(\exp(O((\log n)^{\rc t}(\log \log n)^{1-\rc t})))$, while the best lower bound is $n^{1+\rc{\ce{\fc{q}2+1}}}$ for $q$ queries.
Thus narrowing the gap for $\De(\text{OR}_n,m)$ is a first step towards narrowing the gap for LDC's. 

Secondly, OR representations give explicit constructions of Ramsey graphs, and encompass many previous such constructions~\cite{Gr00, Gr00}. Gopalan defines OR representations slightly differently, as a pair of polynomials $P\pmod p$ and $Q\pmod q$ such that for $\mx\in B^n$, $P(\mx)=0$ and $Q(\mx)=0$ simultaneously only at $\mx=0$. The construction puts an edge between $\mx,\my\in B^n$ iff $P(\mx\opl \my)=0$. The probabilistic method gives nonexplicit graphs with $2^n$ vertices with clique number $\om$ and independence number $\al$ at most $(2+o(1))n$; the best OR representations give explicit graphs with $\om,\al\le e^{O(\sqrt{\log n})}$.

Recently, Bhomwick and Lovett \cite{BL} showed a barrier to lower bounds for the weak degree of $\text{OR}_n$: to prove strong lower bounds, one has to use properties of polynomials that are not shared by \emph{nonclassical polynomials}, because there exist nonclassical polynomials of degree $O(\log n)$ that represent $\text{OR}_n$. 
A nonclassical polynomial of degree $d$ is a function $f:\F_p^n\to \R/\Z$ such that $\De_{\mathbf h_1}\cdots \De_{\mathbf h_{d+1}}f = 0$ for all $\mathbf h_1,\ldots, \mathbf h_{d+1}\in \F_p^n$, where $\De_{\mathbf h}f(\mx):=f(\mx + \mathbf h)-f(\mx)$.  
Thus, to go beyond $\Om(\log n)$,  one cannot rely exclusively on the fact that the $d$th difference of a degree $d$ polynomial is constant, which is the core of techniques such as Weyl differencing. 
This barrier it not directly relevant to our work because nonclassical polynomials for degree $d=2$ can only appear in characteristic 2, and any such nonclassical polynomial $f:\F_2^n\to \R/\Z$ can be realized as a polynomial modulo 4, $4f:\Z_4^n\to \Z_4$.

The maximum $n$ such that a degree 2 polynomial can weakly represent $\text{OR}_n$ is not known. The best symmetric polynomial has $n=8$, but the true answer lies in the interval $[10, 20]$ \cite{TB}, as the polynomial $\pa{\sum_{i=1}^{10}x_i}+ 5(x_1x_{10}+x_2x_9+x_3x_8+x_4x_7 +x_5x_6)$ works for $n=10$.

\subsection{Notation} 
We use the following notation.
\begin{itemize}
\item
$B=\{0,1\}$. Note that we regard $B$ as a subset of $\Z$, hence distinguishing it from $\F_2$.
\item
Boldface font represents vectors; for instance $\mx \in B^n$ is the vector $(x_1,\ldots, x_n)$.
\item $\Z_m$ is the ring of integers modulo $m$.
\item
For $q=p^\al$ a prime power, write $q||m$ ($q$ fully divides $m$) to mean that $p^\al \mid m$ but $p^{\al+1}\nmid q$.
\item
Let $e_m(j) = e^{\fc{2\pi i j}{m}}$. Note this is well defined on $\Z_m$.
\end{itemize}

\section*{Acknowledgements}

Thanks to Zeev Dvir for his guidance and comments on this paper, and to Sivakanth Gopi for useful discussions. 

\section{Proof overview}
\label{sec:over}

It suffices to show that if $n>m^{Cd}$ and $f$ is a quadratic polynomial modulo $m$, then the number of zeros of $f$ is either 0 or $\ge 2$.

We first define the notion of \emph{boolean rank} (Definition~\ref{eq:1rank}). We say a quadratic $f$ has boolean rank at most $r$ if on the Boolean cube, it can be written as a function of $r$ linear forms. Boolean rank is useful because low boolean rank implies $f$ has many zeros, as we will show in Section~\ref{sec:lowrank}.
This is because if $f$ has low boolean rank, then $f(\mx)=0$ whenever $\mx$ solves a small system of linear equations modulo $m$. For example, if $f(\mx)=l_1(\mx)^2+l_2(\mx)^2$, then any solution to $l_1(\mx)=l_2(\mx)=0$ is a solution to $f(\mx)=0$.
Because we have reduced the problem to a linear problem, additive combinatorics comes into play. We use bounds on the Davenport constant \cite{GG} to show that there are many solutions. 

The difficult case is when $f$ has large boolean rank. In Section~\ref{sec:dich}, we show that roughly speaking, this implies $f$ is equidistributed (Theorem~\ref{thm:rigid}). 
Using orthogonality of characters, the fact that for $y\in \Z_m$, 
\[
\rc m \sum_{j\md m}e_m(jy) = \begin{cases}
0,&y\ne 0\\
1,&y=0
\end{cases}
\]
for any function $f:B^n\to \Z_m$, we can count the number of zeros of $f$ using the following exponential sum. (For a similar application of exponential sums in complexity theory, see \cite{Bo}.)
\begin{align}
|\set{\mx\in B^n}{f(\mx)=0}| &= \sum_{\mx\in B^n}\rc m \sum_{j\md m} e_m(jf(x))\\
\implies 
\rc{2^n}|\set{\mx\in B^n}{f(\mx)=0}|
&=\rc m + \rc m \sum_{j\nequiv 0\md m}  \EE_{\mx\in B^n} e_m(jf(x))\label{eq:ct}
\end{align}
If each exponential sum $\E_{\mx\in B^n} e_m(jf(x))$ is small, then the proportion of zeros approximately equals $\rc m$.
We show that high boolean rank implies that these sums are small.

A standard technique to bound an exponential sum is by Weyl differencing: squaring the sum effectively reduces the degree of $f$. Complications arise due to the fact that we are working in $B^n$ rather than the group $\F_2^n$. We will find that the sum is small when the matrix $A_f$ corresponding to $f$ has an off-diagonal submatrix of high rank (\eqref{eq:nonz} and Lemma~\ref{lem:wt}). We show that high boolean rank is equivalent to $A_f$ having high \emph{diagonal rigidity} (Proposition~\ref{pr:rr2}), which in turn implies that $f$ has such a off-diagonal submatrix of high rank (Lemma \ref{lem:rigid}), as desired. Note that diagonal rigidity is a special case of the widely studied notion of matrix rigidity due to Valiant~\cite{Val}.

Finally, we note two technical points. First, we need to define a notion of rank over $\Z_{p^\al}$. We collect the relevant definitions and facts in Appendix~\ref{sec:app}. This makes the proof more technical. For simplicity, the reader may consider the case when $m$ is a product of distinct primes, so that the usual notion of rank over $\F_p$ suffices. 

Secondly, note that if $f$ is already biased modulo $m_1$ for some $m_1\mid m$, then we expect~\eqref{eq:ct} to be biased as well. Thus we factor $m=m_1m_2$ and break the sum in~\eqref{eq:ct} up into $j\nequiv 0\pmod{m_1}$ and $j\equiv 0\pmod{m_1}$. Consider moving prime factors from $m_1$ to $m_2$. If the boolean rank increases slowly at each step, then the boolean rank modulo the ``worst'' prime is bounded, and we are in the low rank case. If the boolean rank increases too fast at any step, we will be in the high rank case. We conclude the theorem in this fashion in Section~\ref{sec:proof}.

\section{Low rank quadratic polynomials have many solutions}
\label{sec:lowrank}
\begin{df}\label{df:1rank}
The \textbf{rank} $\rank(f)$ of a quadratic polynomial $f$ modulo $m$ is the minimal $r$ such that there exists a function $F:\Z_m^r\to \Z_m$ and vectors $\mv_1,\ldots, \mv_r\in \Z_m^n$ such that for all $\mx \in \Z_m^n$,
\begin{equation}\label{eq:1rank}
f(\mx) = F(\mv_1^T\mx,\ldots, \mv_r^T\mx).
\eeq
Note this extends the definition of rank of a quadratic form (the homogeneous case).

The \vocab{boolean rank} $\text{brank}(f)$ is defined the same way, except that~\eqref{eq:1rank} only has to hold for $\mx\in B^n$.
\end{df}
Note that $F$ in Definition~\ref{df:1rank} has a special form here: it is a sum of squares with coefficients. However, we will not use the structure of $F$ in our arguments.

\begin{thm}\label{thm:2soln}
Let $f:B^n\to \Z_m$ be a quadratic polynomial modulo $m$. 
Suppose that for each prime power $q|| m$, $f\pmod q$ has boolean rank $r_q$. Let $r=\sum_{q||m}r_q$. If
$f(\mx)=0$ has a solution $\mx\in B^n$, then the following hold.
\begin{enumerate}
\item
If $n\ge mr\log m$ then $f$ has at least 2 solutions.
\item
$f$ has at least
\[
2^{n - mr \log m\log n}
\]
solutions in $B^n$. 
\end{enumerate}
\end{thm}

The theorem will be a consequence of the following.
\begin{thm}\label{thm:2soln2}
Let $\{\mv_{pi}\in (\Z_q)^n\}_{1\le i\le r_q, q|| m}$ be a collection of $r=\sum_{q|| m}r_q$ vectors.
Then the number of solutions to the system
\[
\mv_{pi}^T\mx =0,\qquad 1\le i\le r_q, q|| m
\]
in $B^m$ 
is at least 2 if $n\ge mr\log m$, and is at least $2^{n-mr\log m \log n}$.
\end{thm}

The proof of this relies on a well-studied problem in additive combinatorics, that of determining the Davenport constant of a group. See \cite{GG} for a survey.
\begin{df}
Let $G$ be an abelian group. The \vocab{Davenport constant} of $G$, denoted $d(G)$ is the 
minimal $d$ such that for all $n>d$ and all $g_1,\ldots, g_n\in G$, the equation
\[
\sum_{i=1}^n x_ig_i = 0
\]
has a nontrivial solution $\mx \in B^n\bs\{0^n\}$.
\end{df}
\begin{thm}[{\cite[Theorem 3.6]{GG}}]
\label{thm:davenport}
Let $G$ be a nontrivial abelian group with exponent $m$. Then 
\[
d(G) \le (m-1) + m\log \fc{|G|}{m}.
\]
\end{thm}

We need to turn this existence result into a lower bound on the number of solutions.
\begin{lem}\label{lem:dav-num}
Let $G$ be a nontrivial abelian group.
The number of solutions $\mx \in B^n$ to 
\[
\sum_{i=1}^n x_i g_i=0
\]
is at least
\[
2^{n-(d(G)+1)\log n}.
\]
\end{lem}
\begin{proof}
Given a solution $\mathbf x_0$, we can apply the definition of $d(G)$ to $\mx - \mx_0$. Hence we see that any $(d(G)+1)$-dimensional slice of $B^n$ that has 1 solution must have another solution. 

Now we claim that every Hamming ball of radius $d(G)$ must have at least 1 solution.  Consider a point $\my$. Take a point $\mx$ solving the equation such that $d(\mx,\my)$ is minimal. If $d(\mx,\my)\ge d(G)+1$, then consider the $d(G)+1$-dimensional slice of $B^n$ that contains $\mx$ and such that moving in any of the $d(G)+1$ directions brings $\mx$ closer to $\my$. There must be another point in this hypercube that solves the equation, contradicting the minimality of $\mx$.

Every Hamming ball of radius $d(G)$ has at least 1 solution, so by counting in two ways, the number of solutions is at least $ \rc{\sum_{k=0}^{d(G)} \binom{n}{k}}2^n = 2^{n-(d(G)+1)\log n}$.
\end{proof}

\begin{proof}[Proof of Theorem~\ref{thm:2soln2}]
This is exactly the equation in the definition of the Davenport constant, where $G=\prod_{q|| m}(\Z_q)^{r_q}$ and $g_i = (v_{qi}^Te_i)_{1\le i\le r_q,q|| m}$. 
The Davenport constant satisfies
\[
d(G)\le  (m-1)+m\log \fc{|G|}m < mr \log m - 1.
\]
Now apply Lemma~\ref{lem:dav-num}.
\end{proof}
\begin{proof}[Proof of Theorem~\ref{thm:2soln}]
By definition of boolean rank there exist $\mv_1,\ldots, \mv_r\in \Z_m^n$ such that for all $\mx \in \Z_m^n$,
\[
f(\mx) = F(\mv_1^T\mx,\ldots, \mv_r^T\mx).
\]
\Wog, $F(\mathbf 0)=0$, so that $f(\mx)=0$ whenever $\mv_1^T\mx=\cdots= \mv_r^T\mx=0$. Now use Theorem~\ref{thm:2soln2}
\end{proof}

\section{High rank implies equidistribution}
\label{sec:dich}

In this section we prove the following theorem.
\begin{thm}[High rank implies equidistribution]\label{thm:rigid}
Let $m>1$ be a positive integer. 
Let $f\in \Z_m[\mx]$ be a quadratic polynomial in $n$ variables.  If there exists a factor $q|| m$ such that $f$ modulo $q$ has boolean rank at least $\Om(m^2\log \prc{\ep})$, 
then 
\[
\ab{
\EE_{\mx\in B^n}e_m(f(\mx))
}<\ep.\]
\end{thm}


First we give a different interpretation for the (boolean) rank. For simplicity, suppose $m=p$ is prime. 
The boolean rank does not change if $f$ changes by a constant, so assume $f$ has constant term 0. For any linear form $f_0$, on $B^n$ we can treat $f+f_0$ as a quadratic form because if $\mx \in B^n$, then $x_i=x_i^2$. 
Hence,
\[
\text{brank}(f)\le 1+ \min_{f_0\text{ linear}}\rank(f+f_0).
\]

Equivalently, when $p\ne 2$, we can think in terms of the matrix $A_f$ corresponding to $f$. Here $A_f$ is the matrix such that $f(\mx) = \mx^T A_f \mx$, i.e., the matrix of the bilinear form $\rc2[f(\mx+\my) - f(\mx)-f(\my)]$. 
By using $x_i=x_i^2$, we have that linear forms $f_0$ corresponds to a diagonal matrices, so 
\[
\text{brank}(f)\le 1 + \min_{D\text{ diagonal}}\rank (A_f + D).
\]
This motivates the following definition. (For the definition of matrix rank when $m$ is composite, see Appendix~\ref{sec:app}.)
\begin{df}\label{df:rigid}
Let $A$ be a matrix over $\Z_m$. We say $A$ is \vocab{$r$-diagonal rigid} if for all diagonal matrices $D$, $\rank (A+D)\ge r$.
\end{df}
Diagonal rigidity is related to a more widely studied notion of matrix rigidity, in which the matrix $D$ can be any sparse matrix. Matrix rigidity is an extensively studied problem with many applications to complexity theory. (See~\cite{rigid} for a survey.) 

We formalize our argument above as the following proposition. The argument extends to prime powers because it still holds that a quadratic form $f$ depends only on the projection of $\mx$ in $\rank(A_f)$ directions (Proposition~\ref{pr:rr}).
\begin{pr}\label{pr:rr2}
Let $m$ be a prime power, $f$ a quadratic polynomial. If $2\mid m$, assume $f$ has even coefficients.
If $A_f$ is $r$-rigid, then
\[
\text{brank}(f)\le r+1.
\]
\end{pr}

Before we prove Theorem~\ref{thm:rigid}, we need a few lemmas.
\begin{lem}\label{lem:linear-sum}
Let $m$ be a positive integer and let $f:B^n\to \Z_m$ be given by a linear polynomial modulo $m$ involving $t$ variables:
\[f(\mx) = \sum_{j=1}^t a_jx_{i_j},a_j\ne 0.\] 
Then
\[
\ab{\EE_{\mx\in B^n} e_m(f(\mx))}\le \pa{1-\rc{m^2}}^t\le e^{-\fc{t}{m^2}}.
\]
\end{lem}
\begin{proof}
The sum decomposes as a product over the coordinates:
\bal
\ab{\E_{\mx\in B^n} e_m(f(\mx))}
&\le \ab{\prod_{j\in [n]} \EE_{x_j\in B} (e_m(a_{i_j}x_j))}\\
&=\prod_{j\in [n]} \af{1+e_m(a_{i_j})}{2}\\
&\le\prod_{j\in [n]} \af{1+e_m(1)}{2}\\
&\le \pa{1-\rc{m^2}}^t.
\end{align*}
In the last step we use $\af{1+e_m(1)}{2} = \cos\pf{\pi}m\le 1-\rc{m^2}$.
\end{proof}

Next we show that a symmetric, rigid matrix has a large off-diagonal submatrix of full rank. The main technicality comes from working over composite moduli.
\begin{lem}\label{lem:rigid}
Let $A$ be a matrix over $\Z_m$, where $m=p^\al$ is a prime power. 

Suppose $A$ is symmetric and $r$-rigid, $r\ge 6$. Then there exist disjoint sets of indices $I_1,I_2$ such that $A_{I_1\times I_2}$ is a square matrix of full rank, with rank at least $\rc{4}r$. 
\end{lem}
\begin{proof}
Suppose $A$ is a $n\times n$ matrix. 

If there are disjoint $I_1$, $I_2$ such that $A_{I_1\times I_2}$ has rank at least $\rc4r$, then the result follows because we can find a square submatrix of full rank (Proposition~\ref{pr:fullrank}).

We show the contrapositive: if the maximum rank of an off-diagonal submatrix is $s\ge1$, then there exists a diagonal matrix $D$ so that $\rank(A+D)\le 4s$. 

Take the off-diagonal matrix of maximal rank. To break ties, choose the matrix whose rows generate the largest subgroup. By Proposition~\ref{pr:fullrank} there is a submatrix whose rows and columns generate an isomorphic subgroup.
Without loss of generality, assume that it has row indices $I_1'=[1,s]$ and column indices $I_2'=[\ff n2 + 1, \ff n2+s+1]$. The matrix $A_{I_1\times I_2}$, $I_1=[1,\ff n2], I_2=[\ff n2 + 1,n]$ also rank $s$.

Now we show that we can pick the first $\ff n2$ entries of $D$ so that $(A+D)_{[1,\ff n2]\times [1,n]}$ has rank at most $2s$. We will also be able to carry out the same procedure on the last $\ce{\frac n2}$ rows by considering the reflection of $A_{I_1'\times I_2'}$ across the diagonal, giving the total of $4s$.

For $s+1\le t\le \fc n2$, consider the matrix $A_{[1,s]\cup \{t\}\times [\ff n2+1,n]\bs \{t\}}$.  
Let $\mv_1,\ldots, \mv_s,\mv_t$ be its rows.
Of all off-diagonal rank-$s$ matrices, $A_{I_1'\times I_2'}$ generates the largest subgroup. Now $A_{[1,s]\cup \{t\}\times [s+1,n]\bs \{t\}}$ contains this matrix so its $t$th row is a linear combination of the previous rows,
\begin{equation}\label{eq:ln}
\mv_t = \sum_i a_i \mv_i.
\eeq
Let us be more precise: The set of $\textbf{a}$ that satisfy~\eqref{eq:ln} is $\mathbf{a}_t + (\text{lnull}(A_{I_1\times I_2}),0) \in \Z_m^s\times \Z_m$ where lnull denotes the left nullspace and $\mathbf{a}_t$ is a particular solution to~\eqref{eq:ln}. In other words,
\begin{equation}\label{eq:lnull}
\text{lnull}(A_{[1,s]\cup \{t\}\times [s+1,n]\bs \{t\}}) = (\text{lnull}(A_{I_1\times I_2}), 0) + \an{(\mathbf{a}_t, -1)}
\subeq \Z_m^s \times \Z_m
\eeq

Now add in the $t$th column: consider the matrix $(D+A)_{[1,s]\cup \{t\}\times [s+1,n]}$. Choose $D_{tt}$ so that
\[
(D+A)_{tt} = \sum_{i=1}^s a_i A_{it}.
\]

Choosing $D_{tt}$ in this way for $s<t\le \ff n2$, we find that the left nullspace of $(D+A)_{[1,\ff n2]\times [s+1,n]}$ is generated by
\begin{align*}
(\text{lnull}(A_{I_1\times I_2}),& 0, \ldots , 0)\\
(\mathbf{a}_{s+1},& -1, \ldots, 0)\\
\vdots&\\
(\mathbf{a}_{\ff n2},& 0, \ldots, -1),
\end{align*}
and hence isomorphic to $\text{lnull}(A_{I_1\times I_2})\times \Z_m^{\ff n2 - s}$. Thus as groups,
\bal
\text{rowspace}((D+A)_{[1,\ff n2]\times [s+1,n]}) 
&\cong \Z_m^{\ff n2} / \text{lnull}((D+A)_{[1,\ff n2]\times [s+1,n]}) \\
&\cong \Z_m^{\ff n2} / \text{lnull}(A_{I_1\times I_2})\times \Z_m^{\ff n2 - s}\\
&\cong \Z_m^s / \text{lnull}(A_{I_1\times I_2}) \cong \text{rowspace}(A_{I_1\times I_2}).
\end{align*}
Hence
\[
\rank ((D+A)_{[1,\ff n2]\times [s+1,n]})  = \rank(A_{I_1\times I_2})=s,
\]
as needed.

Finally, for any choice of $D_{ii},1\le i\le s$, $(D+A)_{[1,\fc n2]\times [1,n]}$ has rank $\le 2s$. This completes the proof.
\end{proof}

\begin{proof}[Proof of Theorem~\ref{thm:rigid}]
By Proposition~\ref{pr:rr2}, a lower bound for the boolean rank gives a lower bound for the rigidity of $A_f$. If $q$ is a power of 2 and $f$ has odd coefficients, then $A_f$ is not well defined. In this case we can replace $m$ by $2m$ and $f$ by $2f$. This neither changes the boolean rank nor the exponential sum. Hence we can assume $A_f$ is $\Om(m^2\log \prc{\ep})$-rigid over $\Z_q$.

We use Weyl's differencing technique. To bound the exponential sum we square it to reduce the degree of the polynomial in the exponent. We have to be careful of the fact that we are working in $B^n$ rather than $\F_2^n$, so the differences are not allowed to ``wrap around." For a function $f$ defined on $B^n$, and $\mathbf h\in \{-1,0,1\}^n$, define
\[
\De_{\mathbf h}f(\mx) = f(\mx+\mathbf h)-f(\mx)
\]
when $\mx+\mathbf h\in B^n$.

We have
\begin{align}
\ab{
\EE_{\mx\in B^n}e_m(f(\mx))
}^{2} &= \rc{2^{2n}} \sum_{\mx,\my\in B^n} e_m(f(\my)-f(\mx))\\
&=\rc{2^{2n}}\sum_{\mathbf h\in \{-1,0,1\}^n} \sum_{x_i={\tiny \begin{cases}
0,&h_i=1\\
1,&h_i=-1
\end{cases}}} e_m(\De_{\mathbf h}f(x))\\
&\le \rc{2^{2n}}\sum_{\mathbf h\in \{-1,0,1\}^n}
\ab{ \sum_{x_i={\tiny\begin{cases}
0,&h_i=1\\
1,&h_i=-1
\end{cases}}} e_m(\De_{\mathbf h}f(x))}
\label{eq:weyl-quad}
\end{align}
Here we used the fact that the set of pairs $(\mx,\my)\in B^n\times B^n$ is the same as the set of pairs $(\mx,\mx+\mathbf h)$ where $\mx,\mathbf h$ satisfy the conditions below the sum.

Let $\Supp(\mathbf h)$ be the set of nonzero entries of $\mathbf h$ and $\ve{\mathbf h}_0:=|\Supp(\mathbf h)|$ be the number of nonzero entries of $\mathbf h$. 
Let $N_{\mathbf h}$ denote the number of nonzero (nonconstant) coefficients of the linear function $\De_{\mathbf h} f$ restricted to subcube of $\mx$ such that $x_i=\begin{cases}
0,&h_i=1\\
1,&h_i=-1
\end{cases}$; note that this subcube is of size $2^{n-\ve{\mathbf h}_0}$.
By Lemma~\ref{lem:linear-sum} the exponential sum is at most
\bal
\ab{\EE_{\mx\in B^n}e_m(f(\mx))}^2 
&\le \rc{2^{2n}} 
\sum_{\mathbf h\in \{-1,0,1\}^n}
2^{n-|\Supp(\mathbf h)|} e^{-N_{\mathbf h}/m^2}\\
&=\sum_{\mathbf h\in \{-1,0,1\}^n} \Pj(\mathbf h) e^{-N_{\mathbf h}/m^2}
\end{align*}
where in the last expression we think of $\mathbf h$ as a random variable with $\Pj(h_i=0)=\rc2$, $\Pj(h_i=\pm 1)=\rc4$.

We show that if $A_f$ is $Cm^2\log \prc{\ep}$-rigid mod $p$, then with high probability $N_{\mathbf h}$ is large, so that $e^{-N_{\mathbf h}/m^2}$ is small.

Note that $N_h$ can be computed as follows. We have that $\De_{\mathbf h} f(\mx) = \mx^T A_f\mathbf h$. Since we are considering the restriction of $\De_{\mathbf h} f$ to a subcube where only the $x_i$ with $i\nin \Supp(\mathbf h)$ are free, $N_{\mathbf h}$ is the number of nonzero entries in $((A_f)\mathbf h)_{[n]\bs \Supp(\mathbf h)}$. 
We can consider choosing $\mathbf h$ in 2 stages. First choose a random partition $I_1\sqcup I_2=[n]$; $I_1$ will contain the indices where $\mathbf h$ is 0 and $I_2$ will contain the indices where $\mathbf h$ is $\pm1$. Then choose $\mathbf h_{I_2}\in \{-1,1\}^{I_2}$ uniformly at random. Now 
\[
(A_f\mathbf h)_{[n]\bs \Supp(\mathbf h)} = \ve{(A_f)_{I_1\times I_2} \mathbf h_{I_2}}_0
\]
so the expected value is 
\begin{equation}\label{eq:nonz}
\ab{\EE_{\mx\in B^n}e_m(f(\mx))}^2 
\le
\EE_{I_1\sqcup I_2=[n], \mathbf h_{I_2}\in \{\pm 1\}^{I_2}} 
\exp(-\ve{(A_f)_{I_1\times I_2} \mathbf h_{I_2}}_0/m^2). 
\eeq

We need the following claim.
\begin{lem}\label{lem:wt}
Suppose that $A$ is a matrix over $\Z_m, m=p^\al$ with rank $r$. Suppose that $\mathbf v\in B^l$  is given and $\mathbf w\in B^k$ is chosen uniformly at random. Let $0\le t\le 1$. Then
\[
\Pj(\ve{\mathbf v+A\mathbf w}_0\le tr) \le 
2^{(t+H(t)-1)r+o(1)}
\]
as $r\to \iy$, 
where $H(t)=-t\lg t-(1-t)\lg (1-t)$.
\end{lem}
\begin{proof}
We may reduce to the case where $A$ has $r$ rows by Proposition~\ref{pr:fullrank}, because having at most $d$ nonzero entries in a given set of $r$ entries is a weaker condition than having at most $d$ nonzero entries.

First we claim that 
for any $d$-dimensional hyperplane $H$, the number of solutions to $\mathbf v +A\mathbf w\in H$ is at most $2^d$. 
Suppose the column space of $A$ is isomorphic to $\prod_{i=1}^r (\Z_{\fc m{a_i}})$. There exists an invertible matrix $M$ such that $D:=MA = \diag(a_1,\ldots, a_r)$.
We have are interested in solutions $\mathbf w\in B^n$ to 
\bal
\mathbf v + A \mathbf w &\in H\\
\iff 
A \mathbf w &\in -\mathbf v + H\\
\iff D \mathbf w &\in -M\mathbf v + MH.
\end{align*}
Let $N$ be a $r\times d$ matrix whose columns generate $H$. We would like to count the number of solutions $\mathbf u$ to 
\bal
D\mathbf w &= - M \mathbf v + M(N\mathbf u)\\
\iff \forall i,\quad 0 \text{ or }a_i &= (- M \mathbf v + M(N\mathbf u))_i
\end{align*}
By putting $MN$ in ``column-echelon form," we find that there are at most $2^d$ possibilities for $\mathbf u$. This proves the claim.

Now note the set $\ve{v+Aw}_0\le d$ is defined by $\binom rd$ hyperplanes.  Thus
\[
|\set{w\in B^r}{\ve{v+Aw}_0\le tr}| \le \binom{r}{tr}2^{tr} =2^{(H(t)+t)r+o(1)},
\]
giving the bound.
\end{proof}

Let $q$ be a prime power fully dividing $m$, and 
suppose $A_f\pmod q$ is $r$-rigid for $r=Cm^2\log \prc{\ep}$, $C$ to be chosen. By Lemma~\ref{lem:rigid}, there exist disjoint $J_1,J_2$ such that $H_{J_1\times J_2}$ has rank $\ge \fc{r}4$ and is full rank.

Let $\de$ be a small constant. We have the following with high probability.
\begin{enumerate}
\item
If $J_1'\subeq J_1$ and $J_2'\subeq J_2$ are random subsets, where each element is included individually with probability $\rc2$, with high probability $|J_1'|\ge (1-\de)\fc r8$ and \[\rank((A_f)_{J_1'\times J_2})\ge (1-\de)\fc{r}8.\]
The probability of failure is $\le \exp(-\fc r8 \de^2)=\exp(-\Om(r\de^2))$.
\item
If item 1 holds, choose any $(1-\de)\fc r8$ columns of $H_{J_1'\times J_2}$ that generate a rank $(1-\de)\fc{r}{8}$ subgroup. With high probability, $J_2'$ will intersect at least $(1-\de)^2\fc r{16}$ of them, and 
\[
\rank((A_f)_{J_1'\times J_2'})\ge (1-\de)^2\fc{r}{16}
\]
The probability of failure is again $\le \exp(-\Om(r\de^2))$.
\item
For $I_1\sqcup I_2=I$ a random partition, the intersections $I_1\cap J_1$, $I_2\cap J_2$ are random, so they can be modeled by $J_1',J_2'$ and we get
\[
\rank((A_f)_{I_1'\times I_2'})\ge (1-\de)^2\fc{r}{16}
\]
By Lemma~\ref{lem:wt}, $\Pj\pa{\ve{(Af)_{I_1'\times I_2'}}_0\le (1-\de)^2\fc{r}{64}}\le 2^{(\rc4+H\prc 4-1)r + o(1)}$, i.e., with high probability
\[
\ve{(A_f)_{I_1'\times I_2'}h}_0 > (1-\de)^2 \fc{r}{64}.
\]
The probability of failure is $\exp(-\Om(r))$.
\end{enumerate}
Thus separating out the terms in the sum which have 
$\ve{(A_f)_{I_1'\times I_2'}\mathbf h}_0 \ge \fc{r}{100}$ in~\eqref{eq:nonz}, we get
\begin{equation}\label{eq:final}
\ab{\EE_{\mx\in B^n}e_m(f(\mx))}^2 \le 
\EE_{I_1\sqcup I_2=[n], \mathbf h_{I_2}\in \{\pm 1\}^{I_2}} 
\exp(-\ve{(A_f)_{I_1\times I_2} \mathbf h_{I_2}}_0/m^2)
\le e^{-\Om(r)} + e^{-\fc{r/100}{m^2}}.
\eeq
In our setting $r=\Om(m^2\log\prc{\ep})$, so~\eqref{eq:final} equals $\ep^2$. This proves the theorem.
\end{proof}

\section{Proof of main theorem}
\label{sec:proof}
\begin{proof}[Proof of Theorem~\ref{thm:main}]
Note that if $m=m_1m_2$ (not necessarily relatively prime) and the proportion of zeros $\rc{2^n}|\set{\mx \in B^n}{f(x)\equiv 0\pmod{m_2}}|$ is already biased, we expect~\eqref{eq:ct} to be biased as well. To take this into account, we separate out the terms where $j\equiv 0 \pmod{m_1}$ and use $e_m(m_1k)=e_{m_2}(k)$. Then~\eqref{eq:ct} becomes
\begin{align}
\eqref{eq:ct}
&=\rc m + \rc m \sum_{j\md m \nequiv 0\md{m_1}}  \EE_{\mx\in B^n} e_m(jf(\mx)) + \rc m \sum_{j\md m\equiv 0\md{m_1}}e_{m}(jf(\mx))\\
&=\rc m + \rc m \sum_{j\md m \nequiv 0\md{m_1}}  \EE_{\mx\in B^n} e_m(jf(\mx)) + \rc m \sum_{k\nequiv 0\md{m_2}}e_{m_2}(kf(\mx))\\
&=\pa{\rc m \sum_{j\md m \nequiv 0\md{m_1}}  \EE_{\mx\in B^n} e_m(jf(\mx))} + \rc{m_1}\pa{\rc{m_2}+\rc{m_2}\sum_{k\nequiv 0\md{m_2}}e_{m_2}(kf(\mx))}\\
&=\pa{\rc m \sum_{j\md m \nequiv 0\md{m_1}}  \EE_{\mx\in B^n} e_m(jf(\mx))} + \rc{m_12^n} |\set{\mx\in B^n}{f(\mx)\equiv 0 \md{m_2}}|.
\label{eq:split}
\end{align}

Let the prime factorization of $m$ be $p_1^{a_1}\cdots p_d^{a_d}$. For $1\le i\le d, 1\le b\le a_i$, let the boolean rank of $f$ modulo $p_i^{b}$ be $r_{i,b}$. (Note that $r_{i,1}\le r_{i,2}\le \cdots$.)
Let $r_1\ge \cdots \ge r_{d'}$ be the numbers $r_{i,b}$ in decreasing order, and let $p_1',\ldots, p_{d'}'$ be the associated primes (so $p_i$ appears $a_i$ times in this sequence).
Consider 3 cases. Let $C$ be the constant in Theorem~\ref{thm:rigid}.
\begin{enumerate}
\item
$r_{d'}>Cm^2\log m$. Then $r_i>Cm^2\log m$ for each $i$.
Note that for $0<j<m$ we have
\[
e_m(jf(\mx)) = e_{\fc{m}{\gcd(m,j)}}(j'f)
\]
where $j'$ is invertible. The boolean rank of $j'f$ and $f$ are equal modulo any prime power dividing $\fc{m}{\gcd(m,j)}$. 
By Theorem~\ref{thm:rigid} on $\fc{m}{\gcd(m,j)}$, we have
\[
\ab{\EE_{\mx\in B^n} e_m(jf(\mx))} <\rc{m}.
\]
Thus by~\eqref{eq:ct}, the proportion of zeros is $\ge\rc{m^2}>\rc{2^n}$, and $f$ does not represent OR${}_n$.
\item
There exists $i$ such that $r_i \ge Cm^3 d (\log m)(\log n) r_{i+1}$. Then by~\eqref{eq:split} on $m_1=p_1'\cdots p_i'$ and $m_2=p_{i+1}'\cdots p_{d'}'$, using Theorem~\ref{thm:2soln} to lower-bound the counts,
\begin{align*}
&\rc{2^n}|\set{\mx\in B^n}{f(\mx)=0}|
\\
&\ge\pa{\rc m \sum_{j\md m \nequiv 0\md{m_1}}  \EE_{\mx\in B^n} e_m(jf(x))} + 2^{-m_2(r_{i+1}+\cdots +r_d) \log m_2\log n-\log m_1}.
\end{align*}
In order for this to be $>\rc{2^n}$ (so that$f$ has more than 1 zero), it suffices to have for each $j\pmod m\nequiv 0\pmod{m_1}$,
\begin{equation}\label{eq:ect}
 \EE_{\mx\in B^n} e_m(jf(x)) < 2^{-m_2(r_{i+1}+\cdots +r_d) \log m\log n}.
\eeq

Because $m_1\nmid j$, for some $p$ we have $v_{p}(m_1)>v_{p}(j)$, and $v_{p}\pf{m}{\gcd(m,j)}> v_{p}(m_2)$. The number of $t>i$ such that $p_t'=p$ is $v_{p}(m_2)$, so the $v_{p}\pf{m}{\gcd(m,j)}$th appearance of $p$, counting from $d$ down to 1, is $p_{s}'$ for some $s<i$. Then the rigidity of $\fc{j}{\gcd(j,m)}f$ modulo $p^{v_{p}\pf{m}{\gcd(m,j)}}$ is at least $r_s\ge r_i$. 

By Theorem~\ref{thm:rigid} on $\fc{j}{\gcd(j,m)}f$ modulo $p^{v_{p}\pf{m}{\gcd(m,j)}}$,~\eqref{eq:ect} holds when
\[
r_i \ge m^2\log (2^{m_2(r_{i+1}+\cdots +r_d) \log m\log n}).
\]
It suffices to have
\[
r_i\ge m^3 d r_{i+1}(\log m)(\log n) ,
\]
which is exactly the assumption for this case.
\item
Neither of the first two cases hold. Then the ratio between consecutive $r_i$ is at most $Cm^3d(\log n)(\log m)$, so 
\[
\sum_{i=1}^d r_i \le (Cm^3d(\log n)(\log m))^d
\]
If $n>m^{4d}$, then this quantity is $< \fc{n}{m\log m}$. Thus by Theorem~\ref{thm:2soln}, $f$ has at least 2 zeros, and $f$ does not represent OR.
\end{enumerate}
\end{proof}

\section{Thoughts on higher degree}
\label{sec:thoughts}

The key reason that this argument works for degree 2 polynomials is that two notions of rank coincide---the boolean rank of $f$ and the rigidity of the associated matrix. When the boolean rank is low, we find that $f(\mx)=0$ has many solutions by solving a series of linear equations; when rigidity is high, the exponential sum is small, and we have close to the expected number of solutions. For degree $\ge 3$ we lose this natural criterion for the exponential sum to be small.

The notion of rank can be naturally generalized. The 1-rank is the notion of rank we used.
\begin{df*}[{\cite[Def. 1.5]{GT}}]
Let $d\ge 0$ and let $f:\Z_m^n\to \Z_m$ be a function. The \vocab{degree $d$ rank} $\rank_d(f)$ is the least integer $k\ge 0$ for which there exist polynomials $Q_1,\ldots, Q_k$ of degree $d$ and a function $F$ such that 
\[
f = F(Q_1,\ldots, Q_k).
\]
\end{df*}

We seek an analogue of Theorem~\ref{thm:rigid} for higher degree. A first attempt is to try to use the Bognadov-Viola Lemma, which says that lack of equidistribution implies low rank.
\begin{lem*}[{\cite[Lem. 24]{BV}}]
Let $\de,\si\in (0,1]$. 
If $P$ is a polynomial of degree $d$ over a finite field $\F$ such that \[|\EE_{\mx\in \F^n}e_{\F}(P(x))|\ge \de,\] there exists a function $\wt P$ agreeing with $P$ on $1-\si$ of inputs, such that 
\[
\rank_{d-1}(\wt P)\le \poly\pa{|\F|,\rc{\de},\rc{\si}}.
\]
\end{lem*}
$\wt P$ is a function of the differences of $P$ in certain directions, which have degree $d-1$. 
For us, this lemma is insufficient for two reasons: 
\begin{enumerate}
\item
$P$ only partially agrees with $\wt P$ (it could be that for all $\wt P(\mx)=0$, we have $P(\mx)\ne0$).
\item
We do not expect $\wt P$ to be equidistributed---far from it: enough differences of $P$ are ``sampled" in order for them to ``concentrate" enough to predict the value of $P$.
\end{enumerate}
Green and Tao prove an exact, but ineffective, form of this result. This was later made algorithmic in~\cite{BHT}.
\begin{thm}[{\cite[Thm. 1.7]{GT}}]
Suppose $0\le d<|\F|$. Suppose $P$ is of degree $d$ and $|\E_{\mx\in \F^n}e_{\F}(P(\mx))|\ge \de$. Then $\rank_{d-1}(P)=O_{\F,\de,d}(1)$. 
\end{thm}
If this result carries over to composite moduli, one could hope to make the following argument, illustrated for $d=3$. If the 2-rank is high, then the exponential sum is small, and we are done. If the 2-rank is low, then we can write $f$ in terms of few quadratics, and perhaps we can then use the $d=2$ case on those quadratics $Q_1,\ldots, Q_r$, proving that they achieve they are 0 simultaneously for enough values of $\mx$. However, if this works at all, it seems that the bounds would be enormous.

\printbibliography
\appendix
\section{Linear algebra over $\Z_m$}
\label{sec:app}

We gather some facts about quadratic forms and matrices over $\Z_m$, where $m$ is composite. For background, see~\cite{MH}.
\begin{df}
For an abelian group $G$, define the \vocab{rank} of $G$ to be the minimal $r$ such that there exist $m_1,\ldots, m_r$, with
\[
G\cong (\Z_{m_1})\times \cdots \times (\Z_{m_r}).
\]
Note if $m$ is a prime power, then this representation is unique up to ordering.

Define the \vocab{rank} of a matrix over $\Z_m$ to be the rank of its image (column space). 

Define the \vocab{rank} of a quadratic polynomial $f$ over $\Z_m$ to be the minimal $r$ such that there exists a function $F$ and vectors $\mv_1,\ldots, \mv_r$ such that $f=F(\mv_1^T\mx,\ldots, \mv_r^T\mx)$.
\end{df}
For example, $\smatt 4022$ has rank 2 (``full rank") over $\Z_8$ because the columns generate the subgroup $\Z_4\times \Z_2$; however, it does not generate the whole group.

We note that many facts about rank carry over to abelian groups. Let $A$ be a matrix over $\Z_m$.
\begin{pr}
The subgroup generated by the rows of $A$ is isomorphic to the subgroup generated by the columns of $A$. Thus, the row and column rank of $A$ are equal.
\end{pr}
\begin{proof}
Using elementary (invertible) row and column operations, $A$ can be put into Smith normal form, i.e., diagonalized. For diagonal matrices, the assertion is clear.
\end{proof}

\begin{pr}\label{pr:fullrank}
Let $m=p^\al$ be a prime power.

Suppose $A$ is a matrix over $\Z_m$ with rank $r$. There exists a subset of $r$ rows of $A$ that span the row space of $A$.

Hence, $A$ has a $r\times r$ submatrix of rank $r$ (a ``full rank" submatrix).
\end{pr}
\begin{proof}
We use the fact that if $G$ is a finite abelian $p$-group, then the representation $G=\prod_{i=1}^r (\Z_{p^{\al_i}})$ is unique and the number of factors equals the rank.

Induct on $r$. The claim is true for $r=1$. 
Let $p^a$ be the maximal order of an element in the row space (the order of any element in $\Z_m^k$ is a power of $p$). Because the order of an abelian group is the gcd of the orders of elements in a generating set, there is a row $\mv$ with order $p^a$. Choose this row.

Because $a$ was chosen maximal, the row space is isomorphic to $\an{\mv}\times R'$ for some $R'$ of rank $r-1$. Now consider the projection of the remaining rows to $R'$, and apply the induction hypothesis.

For the last claim, apply the fact to the rows of $A$ and then the columns of the resulting matrix.
\end{proof}

\begin{pr}\label{pr:rr}
Suppose $f$ is a quadratic form over $\Z_{p^\al}$. Let $A_f$ be the associated matrix. If $p=2$, assume that all coefficients of $f$ are divisible by 2, so that $A_f$ is well defined.

Then $\rank(f)=\rank(A_f)$.
\end{pr}
Essentially, the difference between the two is that $\rank(f)$ is the minimal size of a matrix $D$ such that there exist $S$ with $A_f=P^TDP$, while $\rank(A_f)$ is the minimal size of the matrix $D'$ such that there exist $S,T$ with $A_f=P^TD'Q$.
\begin{proof}
From the comment, it is clear that $\rank(f)\ge \rank(A_f)$. Let $D,D'$ be the smallest matrices as above and let $n$ be the size of $D$.  Suppose \bwoc{} that $\rank(D')<n$. Then the left nullspace of $D'$ must contain a subgroup isomorphic to $\Z_{p^\al}$. Take a generator $v_1$ for this subgroup.  Complete $\{v_1\}$ to a generating set $\{v_1,\ldots, v_n\}$ for $\Z_m^n$. From $v_1^TD=0$ and $Dv_1=0$ ($D$ is symmetric) we see that $f$ depends only on $\mv_2^T\mx,\ldots, \mv_n^T\mx$, contradiction.
\end{proof}

%

\end{document}